%% file: paper.tex
\documentclass[runningheads]{llncs}
\usepackage{graphicx}
\graphicspath{{gfx/},{experiments/},}  
\RequirePackage[nolist]{acronym}

\usepackage{subcaption}

\usepackage{listings,algorithm}
\usepackage[noend]{algpseudocode}
\algdef{SE}[DOWHILE]{Do}{doWhile}{\algorithmicdo}[1]{\algorithmicwhile\ #1}%

\usepackage{amsmath, amssymb}

\usepackage[table,dvipsnames]{xcolor}

\usepackage{enumitem}

\usepackage[nospace]{cite}

\usepackage{xfrac,xspace}
\usepackage{relsize}
\usepackage[normalem]{ulem}

\usepackage{tikz}
\usetikzlibrary{positioning}

\usepackage{comment}

\usepackage[pagewise]{lineno}  

\RequirePackage[%
  bookmarks,
  unicode,
  breaklinks=true,
  ]{hyperref}
  
\usepackage[nameinlink]{cleveref}

\newcommand{\splitatcommas}[1]{%
	\begingroup
	\begingroup\lccode`~=`, \lowercase{\endgroup
		\edef~{\mathchar\the\mathcode`, \penalty0 \noexpand\hspace{0pt plus 1em}}%
	}\mathcode`,="8000 #1%
	\endgroup
}

\makeatletter
\renewcommand{\paragraph}{\@startsection{paragraph}{5}{0em}%
  {.7ex plus .2ex minus .1ex}%
  {-.5em}%
  {\bfseries}}
\makeatother

\makeatletter
\def\orcidID#1{\smash{\href{http://orcid.org/#1}{\protect\raisebox{-1.25pt}{\protect\includegraphics{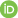}}}}}
\makeatother

\allowdisplaybreaks[4]

\Crefname{equation}{eq.}{eqs.}
\crefname{equation}{equation}{equations}
\Crefname{figure}{Fig.}{Figs.}
\crefname{figure}{figure}{figures}
\Crefname{tabular}{Table}{Tables}
\crefname{tabular}{table}{tables}
\Crefname{definition}{Def.}{Defs.}
\crefname{definition}{definition}{definitions}
\Crefname{proposition}{Prop.}{Props.}
\crefname{proposition}{proposition}{propositions}
\Crefname{section}{Sec.}{Sections}
\crefname{section}{section}{sections}
\Crefname{subsection}{Sec.}{Sections}
\crefname{subsection}{subsection}{subsections}
\crefname{algorithm}{algorithm}{algorithms}
\crefname{listing}{code}{code\ blocks}
\crefformat{footnote}{#2\footnotemark[#1]#3}

\input{defs.tex}

\usepackage[firstpage]{draftwatermark}
\SetWatermarkText{%
\hspace*{3.8in}
\raisebox{5.8in}
{\includegraphics[scale=0.6]{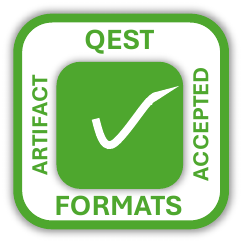}}%
}
\SetWatermarkAngle{0}

\begin{document}
\title{Transient Evaluation of Non-Markovian Models by Stochastic State Classes and Simulation}
\titlerunning{Transient Evaluation of Non-Markovian Models by SSCs and Simulation}

%
%
\author{Gabriel~Dengler\inst{1}\orcidID{0000-0002-4217-4952} \and
Laura~Carnevali\inst{2}\orcidID{0000-0002-5896-4860} \and\newline
Carlos E.~Budde\inst{3}\orcidID{0000-0001-8807-1548}\and
Enrico~Vicario\inst{2}\orcidID{0000-0002-4983-4386}}
\authorrunning{G. Dengler et al.}
%
\institute{Saarland University, Saarbrücken, Germany
\\\email{dengler@depend.uni-saarland.de}
\and
Department of Information Engineering, University of Florence, Florence, Italy
\and
DISI, University of Trento, Trento, Italy
}
\maketitle              

\begin{abstract}
\markchange{Non-Markovian models have} great expressive power, at the cost of complex analysis of the stochastic process.
%
%
The method of Stochastic State Classes~(SSCs) derives closed-form analytical expressions for the joint Probability Density Functions (PDFs) of the active timers with marginal expolynomial PDF,  though being hindered by the number of concurrent non-exponential timers and of discrete events between regenerations.
\markchange{%
Simulation is an alternative capable of handling the large class of PDFs samplable via inverse transform, which however suffers from rare events.}
%
\markchange{We combine} 
these approaches to analyze time-bounded transient properties of non-Markovian models. 
We enumerate SSCs near the root of the state-space tree and then rely on simulation to reach the \markchange{target}, 
%
\markchange{affording} transient evaluation of models for which the method of SSCs is not viable while reducing computational time and variance of the estimator of transient probabilities with respect to simulation.
Promising results are observed in the estimation of rare event probabilities. 


\end{abstract}

\acresetall

\renewcommand{\thelstlisting}{\arabic{lstlisting}}  

\input{01-introduction.tex}

\input{02-sscs.tex}

\input{03-developed-procedure.tex}

\input{04-confidence-intervals.tex}

\input{05-experimental-evaluation.tex}

\input{06-conclusion.tex}


{\paragraph{Acknowledgments.}
We thank Pedro R. D'Argenio for the fruitful discussion that helped us with the derivation of \aclp{CI}.
This work was partially funded by DFG grant 389792660 as part of \href{https://perspicuous-computing.science}{TRR~248 -- CPEC}, the European Union under the INTERREG North Sea project \emph{STORM\_SAFE} of the European Regional Development Fund, the NRRP ``Telecommunications of the Future'' (PE00000001, \emph{RESTART}), MSCA grant 101067199 (\emph{ProSVED}), and the Next\-GenerationEU pro\-jects P2022A492B (\emph{ADVENTURE}) and D53D230084\-00006 (\emph{SMARTITUDE}) under 
the MUR PRIN 2022 PNRR.
}

\paragraph{Data Availability Statement.} A reproduction package for our experiments (i.e.\ the artifact of this paper) is available at \texttt{\color{blue}\href{https://doi.org/10.6084/m9.figshare.25665198}{10.6084/m9.figshare.25665198}}.

\input{95-appendix}

\input{acronyms.tex}

%
%
%
\bibliographystyle{splncs04}
\bibliography{references.bib}

\end{document}

%% file: defs.tex
\newcommand{\vect}[1]{\vec{#1}}

\newcommand{\from}{\colon}
\newcommand{\N}{\mathbb{N}}

\newcommand{\R}{\mathbb{R}}

\newcommand{\Qpos}{\mathbb{Q}_{\geqslant 0}}

\newcommand{\Rpos}{\mathbb{R}_{\geqslant 0}}
\newcommand{\RPos}{\mathbb{R}_{> 0}}

\newcommand{\prop}{\ensuremath{\varphi}\xspace}
\newcommand{\est}{\ensuremath{\hat{p}}\xspace}
\newcommand{\relerr}{\ensuremath{\varepsilon}\xspace}

\newcommand{\age}{\tau_{\mathit{age}}}
\newcommand{\agetau}{{\langle \age,\vect{\tau} \rangle}}
\newcommand{\agetaup}{{\langle \age',\vect{\tau}' \rangle}}

\newcommand{\tstate}{\langle m,D,f_\agetau \rangle}

\newcommand{\tstatep}{\langle m',D',f'_\agetau \rangle}

\newcommand{\F}{F}
\newcommand{\W}{W}
\newcommand{\M}{\mathcal M}

\DeclareMathAlphabet{\mathsc}{OT1}{cmr}{m}{sc}

\newcommand{\dist}[1]{\mathsc{\MakeLowercase{#1}}}
\mathchardef\mhyph="2D 

\let\emptyset\varnothing

\newcommand{\hide}[1]{}
\newcommand{\bfred}[1]{\textsf{\bfseries\color{Red}#1}\xspace}

\colorlet{colcodekey}{PineGreen!60!green!70!black}
\colorlet{colcodeidentifier}{Sepia}
\colorlet{colcodenumeric}{blue!25!black}
\colorlet{colcodesyncaction}{violet!90}
\colorlet{colcodecomment}{Black!60}
\colorlet{coldistribution}{Cerulean!60!Black}
\lstdefinelanguage{Kepler}{  
  keywords={},  
  keywords=[2]{toplevel,and,or,pand,vot,fdep,spare,be,sbe,rbox,seq,csp,wsp,hsp},
  keywords=[3]{fail,repair,dorm,prio,fcfs,rand,lambda,prob,res,rate,mean},
  keywords=[4]{exponential,erlang,uniform,normal,lognormal,%
               weibull,gamma,rayleigh,fsteq,dirac,%
               exp,uni,dir},
  otherkeywords={},
  morecomment=**[l]{//},      
  morecomment=**[s]{/*}{*/},
  moredelim=**[is][\color{colcodenumeric}]{^}{^},     
  moredelim=**[is][\color{colcodekey}]{-|}{|-},       
  moredelim=**[is][\slshape]{__}{__},                 
}
\lstdefinestyle{Kepler}{
  language=Kepler,
  xleftmargin=2em,
  basewidth=0.5em,
  basicstyle={\scriptsize\ttfamily},
  identifierstyle=\color{colcodeidentifier},
  keywordstyle=\bfseries,
  keywordstyle=[2]\color{colcodesyncaction},
  keywordstyle=[3]\bfseries\color{colcodekey},
  keywordstyle=[4]\bfseries\color{coldistribution},
  commentstyle=\color{colcodecomment},
  stringstyle=\mdseries,
  showstringspaces=false,
  numbers=left,
  numberstyle=\scriptsize\ttfamily\color{colcodecomment},
  numbersep=0.9em,
  tabsize=4,
  frame=none,
  aboveskip=\bigskipamount,
  belowskip=\medskipamount,
  abovecaptionskip=\smallskipamount,
  belowcaptionskip=\smallskipamount,
  captionpos=b,
  escapeinside={`}{`},
  mathescape=true,  
}

\newcommand{\markchange}[1]{#1}  

%% file: 01-introduction.tex
\section{Introduction}
\label{sec:intro}

\markchange{Quantitative evaluation of stochastic timed models}
is a difficult problem.
%
%
While the Markovian case counts with time-tested analytical and numerical solutions~\cite{Gra77,uniformization-algorithm,BK08,BP19}, non-Markovian models are much harder to analyze~\cite{supplementary-variables-1, supplementary-variables-2}. 

\paragraph{Motivation.}

\hide{\color{red}Comment from Laura: In Morivation paragraph, I would remark that we address transient analysis of non-Markovian models with multiple concurrent \ac{GEN}, i.e., non-Exponential, timers, given that (i) this is an expressive class of models (capable of representing several application contexts) and that (ii) relevant properties of interest can be expressed by transient rewards (e.g., system reliability, as discussed below). }

Our main research goal is to \emph{quantify time-bounded transient properties of non-Markovian systems}.
Specifically, non-Markovian models with multiple concurrent timers having non-Exponential \ac{GEN} distributions~\cite{stochastic-state-classes} capture characteristics of a large variety of systems, such as real-time systems, cyber-physical systems, and software subject to aging. Notably, they have been used to define quantitative safety and liveness properties of safety-critical sys\-tems---e.g.\ in aerospace, railway, and nuclear industries \cite{Rob00,LM21,performance-ertms-ects-1,ZZZ21}---and to give semantics to RAMS standards---reliability, availability, maintainability, and safety \cite{OK11,RAMS-conference}---including fault tree analysis and reliability block diagrams \cite{faultflow, CCF+21}.

\hide{
Petri nets are a versatile modelling formalism for the task, with multiple concurrent transitions whose firing times can be described by non-Exponential \ac{GEN} timers, each following a continuous arbitrary distribution \cite{gspn-paper,MBB+89}.
These \acp{STPN} have been used to quantify safety and liveness properties of safety-critical sys\-tems---e.g.\ in the aerospace, railway, and nuclear industries \cite{Rob00,LM21,performance-ertms-ects-1,ZZZ21}---and to give semantics to RAMS standards---reliability, availability, maintainability, and safety \cite{OK11,RAMS-conference}---including fault tree analysis and reliability block diagrams \cite{CCF+21,JKSV18}.}

The expressive power of non-Markovian models comes at the cost of  complex analysis of the underlying stochastic process to evaluate the properties of interest. 
A relevant example in RAMS engineering is the evaluation of the \emph{system reliability}, i.e.,~the probability~$p$ to observe an undesired event during mission time, which is a time-bounded safety property \prop to check on the model~\cite{OK11,ZZZ21}.
Numerical algorithms such as \ac{VI} can approximate this quantity for Markovian systems~\cite{FKNP11,BK08,CH08}, via exhaustive explorations of the states of the model~\cite{FKNP11,HHH+19}.
For \ac{GEN} transitions, \ac{VI} needs approximations such as phase-type distributions~\cite{You05,Neu81}.
However, besides approximation errors, this exacerbates the state-space explosion problem, which hinders numeric algorithms like \ac{VI} and ultimately renders them unfit to study non-Markovian models~\cite{FKNP11,Har15}.

\hide{A main concern in RAMS is quantifying \emph{system reliability}: the probability $p$ to observe an undesired event during mission time, which is a time-bounded safety property \prop to check on the \ac{STPN} model \cite{OK11,ZZZ21}.
Numerical algorithms such as \ac{VI} can approximate this quantity for Markovian systems \cite{FKNP11,BK08,CH08}, via exhaustive explorations of the states that make up the \ac{STPN} semantics \cite{FKNP11,HHH+19}.
For \ac{GEN} transitions, \ac{VI} needs approximations such as phase-type distributions \cite{You05,Neu81}.
However, besides approximation errors, this exacerbates the state-space explosion problem that hinders numeric algorithms like \ac{VI}, and which ultimately renders them unfit to study non-Markovian models \cite{FKNP11,Har15}.
}


\paragraph{Related Work.}


Analytical quantification of properties in non-Markovian models is viable only under restrictions on the class of \ac{GEN} distributions, and on the number of concurrent timers in the stochastic process states~\cite{ciardo1994characterization}.
For models subtending a \ac{MRP}~\cite{kulkarni2016modeling},
most approaches address the subclass where up to one \ac{GEN} timer is enabled in each state, i.e.\ the \emph{enabling restriction}~\cite{choi1994markov,german1995transient,amparore2013component}.
The method of supplementary variables~\cite{telek2001transient,supplementary-variables-1} does not require this restriction theoretically, but is impractical without it~\cite{telek2001transient}.
In contrast, sampling the stochastic process at equidistant time points can overcome the enabling restriction~\cite{lindemann1999transient,zimmermann2012modeling}, but requires timers to follow a \ac{DET} or \ac{EXP} distribution. 
The compositional approach of~\cite{CGSV22} does not restrict the number of \ac{GEN} timers either but requires the underlying stochastic process to be decomposable into a hierarchy of \acp{SMP}.

Also, the method of \acp{SSC}~\cite{stochastic-state-classes,HORVATH2012315,sirio-github-code} can address models with multiple concurrent \ac{GEN} timers, provided that a regeneration is always reached in a bounded number of discrete events (the \emph{bounded regeneration restriction}).
\acp{SSC} are restricted to \ac{GEN} timers in the expolynomial class---sum of products of \ac{EXP} and polynomials---which includes \ac{EXP}, uniform, triangular, and Erlang distributions.
The approach derives the closed-form expression of the \ac{PDF} of the active timers after each discrete event and its applicability is hindered by large numbers of concurrent \ac{GEN} timers and of discrete events between regenerations.

Simulation can also quantify a (transient) property \prop on non-Markovian systems \cite{YS02}.
When formal system models such as \acp{STPN} are available, this is called \ac{SMC} \cite{LL16}.
\ac{SMC} can study any stochastic system whose stochastic kernel is known, and from which samples can be drawn, e.g.\ via the inverse transform method \cite{AP18,LL16,ZBC12}.
In contrast to analytical methods like \ac{SSC} analysis, which provide an exact value for the quantity $p$ characterized by \prop, sequential \ac{SMC} generates simulation traces to produce an \emph{estimate} $\langle\est,\relerr\rangle$ s.t.\ $\est\in p\pm\relerr$ with some desired probability $\delta$.
However, this is hindered by \emph{rare events}: when the property \prop to be quantified requires simulation traces to visit states that occur with very low probability, then \relerr or the number of traces explodes, rendering standard simulation useless \cite{RT09b,ZBC12}.

\Ac{RES} tackles such problems, where \ac{ISPLIT} has been used in \ac{SMC} to quantify rare transient properties of non-Markovian systems \cite{BDHS20,LLLT09}.
\ac{ISPLIT} splits the state space ${\mathcal{S}=\uplus_{i=0}^n \mathcal{S}_i}$ to estimate the conditional probabilities $p_i$ of reaching a state in $\mathcal{S}_i$ from $\mathcal{S}_{i-1}$, where the states satisfying \prop are in $\mathcal{S}_n$ and $p=\prod_{i=1}^n p_i$ \cite{LLLT09}.
This works when all the estimates $\est_i$ can be approached via crude \ac{MC}, which rules out rare events caused by single transitions of very low probability, e.g.\ \acp{EXP}
\markchange{%
with a very low rate.
}
Such cases can be tackled by \ac{IS}, which changes the \acp{PDF} $f$ of concurrent timers for a proposed $\tilde{f}$, making it more likely to observe states that satisfy \prop \cite{LMT09}.
An unbiased estimate \est is then obtained by multiplying the result of \ac{MC} by the \emph{likelihood ratio} $\sfrac{f}{\tilde{f}}$.
The drawback of \ac{IS} is that $\tilde{f}$ is problem-dependent, usually defined ad hoc, and bad choices result in worse-than-\ac{MC} convergence \cite{LMT09}.
Automatic $\tilde{f}$ selection is restricted to specific distributions (mainly \ac{EXP}) and model structures \cite{BBS21,RBSJ18}, and even adaptive \ac{IS} approaches such as cross-entropy require non-trivial parameter tuning \cite{BKMR05}.

\paragraph{Contributions.}
In this paper, we evaluate time-bounded transient properties of non-Markovian systems by combining state-space analysis via \acp{SSC} with simulation, capturing rare events while not incurring state-space explosion~\cite{FKNP11,stochastic-state-classes}. 
To the best of our knowledge, this solution has never been attempted in formal \ac{SMC} frameworks~\cite{ABB+24}. 
To approach it, we enumerate \acp{SSC} near the root of the state-space tree and then perform simulation from there on, deriving \acp{CI} for the probability that a property~\prop of interest is satisfied.
Experimental results show that the approach enables evaluation for models for which the method of \acp{SSC} is not viable, notably reducing computational time and variance of the probability estimator with respect to both \ac{MC} simulation and \ac{IS}.
Moreover, promising results are obtained in the estimation of rare event probabilities, opening the way to further research directions.

\hide{Analytical solutions e.g.\ via \acp{SSC} do not suffer from rare events \cite{FKNP11,stochastic-state-classes}.
To overcome their state-space explosion limitations, one could apply them up to a predefined expansion level, and then use simulation to estimate reachability of the event of interest.
To the best of our knowledge this has not yet been attempted in a formal \ac{SMC} framework \cite{ABB+24}---we approach it as follows:
}

In the rest of the paper, first we recall background concepts (\Cref{sec:background}). Then, we present our approach (\Cref{sec:introducing-switch-sscs}) and derive \acp{CI} for properties of interest (\Cref{sec:results-confidence-intervals}). Finally, we present experimental results (\Cref{sec:experimental-evaluation}) and draw conclusions (\Cref{sec:conclusion}).
Additional experimental and implementation details are in \Cref{sec:pn-dft-example,sec:algorithmic-desc-sampling}.


%% file: 02-sscs.tex
\section{Background}\label{sec:background}

In this section, we recall \acp{STPN} (\Cref{subsec:stpn}) as well as aspects of the method of \acp{SSC} (\Cref{subsec:ssc}) and \ac{MC} simulation (\Cref{subsec:mcs}) that are relevant for our work.

\subsection{Stochastic Time Petri Nets}\label{subsec:stpn}

\markchange{\acp{STPN}~\cite{stochastic-state-classes,oris21} model concurrent systems with stochastic durations and discrete probabilistic choices.
As shown in \Cref{fig:stpn}, an \ac{STPN} consists of: 
places with tokens (circles with dots), modeling the discrete logical state; 
transitions (bars) modeling activities with stochastic duration;
directed arcs (directed arrows), from input places to transitions and from transitions to output places, modeling precedence relations among activities.
A transition is enabled by a marking (i.e.,~an assignment of tokens to places) if each of its input places contains at least one token, and its enabling function (``? expression") evaluates to true.} 


Upon enabling, a transition samples a time-to-fire from its \ac{CDF}, i.e.,~\ac{EXP},  \ac{GEN}, or the generalized  \ac{CDF} of a Dirac delta function, 
where transitions with zero time-to-fire are called \ac{IMM}
%
%
\markchange{(in \Cref{fig:stpn}, \ac{IMM} and \ac{GEN} transitions are drawn as thin and thick vertical bars, respectively).}
%
%
The transition with minimum time-to-fire fires, removing one token from each of its input places, adding one token to each of its output places, and applying its update function, i.e.,~an assignment of tokens to each place, defined by a marking expression \markchange{(not present in \Cref{fig:stpn}).}
%
Ties (i.e.,~limit cases of synchronization among DET transitions with the same time-to-fire, e.g., occurring when they are enabled at the initial time) are solved by a random switch determined by probabilistic weights of transitions \markchange{(not present in \Cref{fig:stpn}).}

\begin{figure}[b!]
	\centering
	\includegraphics[width=0.8\textwidth]{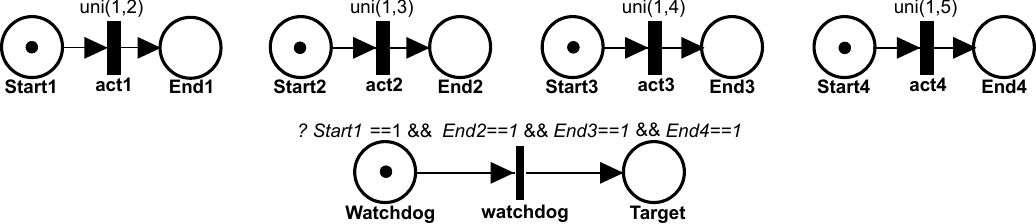}
	\caption{STPN modeling four parallel overlapping activities\label{fig:stpn}}
\end{figure}

\Cref{fig:stpn} shows the STPN of four parallel activities with uniformly distributed duration, modeled by the \ac{GEN} transitions \texttt{act1}, \ldots, \texttt{act4}. The \ac{IMM} transition \texttt{watchdog} fires as soon as activities represented by \texttt{act2}, \texttt{act3}, and \texttt{act4} have been completed while the activity represented by \texttt{act1} is still ongoing.

The SIRIO library~\cite{sirio-github-code} implements syntax and semantics of \acp{STPN} where \ac{GEN} transitions have \textit{expolynomial} \ac{CDF}~\cite{trivedi2009sharpe} (i.e.,~sums of products of exponentials and polynomials), including \ac{EXP}, uniform, triangular, and Erlang \acp{CDF}.

\subsection{Transient Evaluation by the Method of Stochastic State Classes}\label{subsec:ssc}

An \ac{SSC}~\cite{stochastic-state-classes} comprises a marking, a joint support, and a \ac{PDF} for vector~$\agetau$ encoding times-to-fire of the $G$ enabled transitions and the absolute elapsed time (the ``age'' $\age\in\R_{\geq0}$). Given an \ac{SSC}~$\Sigma$, a succession relation provides the joint support and the joint PDF of~$\agetau$ conditioned on the firing of a transition~$\gamma$.

\begin{definition}[\ac{SSC}]\label{def:tsclass}
	An \ac{SSC} is a tuple
	$ \Sigma = \langle m,D_{\langle \age,\vect{\tau}\rangle},f_{\langle \age,\vect{\tau} \rangle}\rangle$ where:
	$m\in \M$ is a marking;
	$f_\agetau$ is the PDF (immediately after the previous firing) of the random vector $\agetau$ including the age timer $\age$ and the times-to-fire $\vect{\tau}$ of transitions enabled by $m$;
	and, $D_{\langle \age,\vect{\tau}\rangle} \subseteq \R^{G+1}$ is the
	support of $f_\agetau$.
\end{definition}

\begin{definition}[Succession relation]\label{def:successor}
	$\Sigma'=\tstatep$ is the successor of $\Sigma=\tstate$ through transition~$\gamma$
	with probability~$\mu$ (i.e.,~$\Sigma\stackrel{\gamma,\mu}{\longrightarrow}\Sigma'$), 
	if, given that the marking is~$m$ and $\agetau$ is distributed over $D$ according to $f_{\agetau}$, then the firing of $\gamma$ has probability $\mu>0$ in $\Sigma$ and yields marking $m'$ and vector of times to fire $\agetaup$ distributed over $D'$ according to
 $f'_{\agetau}$.
\end{definition}
From an initial \ac{SSC} where the times-to-fire of the enabled transitions are independently distributed, the relation $\stackrel{\gamma,\mu}{\longrightarrow}$ can be enumerated by computing firing probabilities and resulting joint supports $D'$ and \acp{PDF}  $f'_{\agetau}$ of vector~$\agetau$:
$D'$ is a \ac{DBM}, i.e.,~solution of a set of linear inequalities constraining the difference between two times-to-fire; 
for models with expolynomial \ac{GEN} transitions,  $f'_{\agetau}$ takes piecewise analytical form (i.e.,~\textit{multivariate expolynomial}) over a partition of $D'$ in DBM sub-zones~\cite{bernstein-polynomials-1}.
Enumeration of $\stackrel{\gamma,\mu}{\longrightarrow}$ yields a transient tree: nodes are \acp{SSC} and edges are labeled with transitions and their firing probabilities~\cite{HORVATH2012315}.
Depending on the number of concurrent non-\ac{EXP} transitions, after a large number of firings, the number of DBM-subzones may significantly increase, leading to a runtime explosion.
%
For \ac{MRP} models under the bounded regeneration restriction, the problem is largely mitigated by enumerating \acp{SSC} between any two regenerations~\cite{HORVATH2012315}, i.e.,~\acp{SSC} where all transitions are newly enabled or enabled by a \ac{DET} time, and thus $f'_{\agetau}$ takes the same analytical representation over the entire domain.
%


Given initial marking~$m_0$ and \ac{PDF}~$f_{\langle \age,\vect{\tau}\rangle}$ for $\langle \age,\vect{\tau}\rangle$, the STPN semantics induces a probability space $\langle \Omega_{m_0}, \mathbb{F}_{\langle \age,\vect{\tau}\rangle}, \mathbb{P}_{m_0, f_{\langle \age,\vect{\tau}\rangle}}\rangle$: $\Omega_{m_0}$ is the set of feasible timed firing sequences and $\mathbb{P}_{m_0, f_{\langle \age,\vect{\tau}\rangle}}$ is a probability measure over them~\cite{paolieri2015probabilistic}. 



An \ac{STPN} identifies a \ac{TPN}~\cite{berthomieu1991modeling,lime2003expressiveness} with same set of outcomes $\Omega_{m_0}$. 
The state $\langle m, \langle \age,\vect{\tau}\rangle \rangle$ of a \ac{TPN} encodes marking $m\in \M$ and vector $\langle \age,\vect{\tau}\rangle$ of the age timer and the times-to-fire of the enabled transitions. 
The state space is covered by \acp{SC}, each \ac{SC} $S = \langle m, D\rangle$ encoding marking~$m$ and joint support~$D$ for~$\langle \age,\vect{\tau}\rangle$.
A reachability relation is defined between \acp{SC}: $S' = \langle m', D'\rangle$ is the successor of $S = \langle m, D\rangle$
\markchange{%
via transition~$t$ if, from marking $m$ and $\langle \age,\vect{\tau}\rangle$ supported over~$D$, $t$ fires in $S$ and yields}
marking $m'$ and vector $\langle \age,\vect{\tau}\rangle'$ supported over~$D'$.
From initial marking~$m_0$ and domain~$D_0$ for~$\langle \age,\vect{\tau}\rangle$, \ac{SC} enumeration yields a \ac{SCG} encoding the set of outcomes~$\Omega_{m_0}$, enabling correctness verification of the TPN.
%

In SSC $\Sigma=\tstate$, an enabled transition~$\gamma$ has null firing probability iff domain $D$ conditioned on $\gamma$ firing first has a non-null measure. Therefore, firings having null probability can be excluded from the \ac{SCG}, which can then be used to determine reachability between \acp{SSC}, i.e.,~\ac{SSC}~$\Sigma'=\tstatep$ is reachable from \ac{SSC}~$\Sigma=\tstate$ iff \ac{SC}~$S'= \langle m', D'\rangle$ underlying $\Sigma'$ is reachable from \ac{SC} $S= \langle m, D\rangle$ underlying $\Sigma$. 
According to this, in the verification of time-bounded transient properties (e.g.,~probability that a marking condition is satisfied by time~$T$),  successor \acp{SSC} are enumerated iff target \acp{SSC} (i.e.,~those satisfying the property of interest) are reachable from them, which can be decided on the \ac{SCG}, as just discussed. 
For very complex models for which the \ac{SCG} enumeration is not viable, the marking graph could likely be enumerated (i.e.,~the graph encoding the reachability relation between markings), so as to avoid computation of the \acp{SSC} from which the target \acp{SSC} cannot be reached regardless of timing constraints.
Alternatively, if the \ac{SCG} not encoding the age variable can be enumerated, it could still be used to detect the \acp{SSC} from which the target \acp{SSC} are not reachable regardless of the elapsed time.

%
%
%


\subsection{Monte Carlo Simulation and Importance Sampling}\label{subsec:mcs}


\ac{MC} simulation performs $n$~independent executions of an \ac{STPN}, estimating the probability~$p_{\prop}(t)$ of marking condition~$\prop$ at time~$t$ as the fraction of executions that satisfy~$\prop$ at time~$t$.
The state of an \ac{STPN} is characterized by a \ac{RV}~$Y$ with support on the state space $\mathcal{S}$, where the samples $y_i$ for $i\in \{1, 2, ..., n\}$ describe a specific state.
$p_{\prop}(t)$ is the mean~$\mu$ of a Bernoulli distributed \ac{RV}~$X = \Psi(Y)$, where $\Psi : \mathcal{S} \rightarrow \{0, 1\}$, whose independent samples $x_i$ for $i\in \{1, 2, ..., n\}$ are equal to either $1$ or $0$ depending on whether $\prop$ is satisfied at time~$t$ of the $i$-th execution or not, respectively.
According to this, the mean $\mu$ and variance $\sigma^2$ of~$X$ can be estimated as the sample mean $\overline{X}_\emph{sim}=\sum_{i=1}^{n}x_i / n$ and the variance $\Tilde{\sigma}_\emph{sim}^2 = \overline{X}_\emph{sim} (1 - \overline{X}_\emph{sim})$, respectively.
When using \ac{IS}~\cite{importance-sampling}, samples are associated with a likelihood $L : \mathcal{S} \rightarrow \mathbb{R}_{\geq 0}$ to compensate for the change of \ac{PDF}. Here, we estimate the mean $\mu$ of the reward with the sample mean $\overline{X}_\emph{IS} = \sum_{i = 1}^{n} L(y_{i})\Psi(y_{i})/n$ and the variance $\sigma^2$ with the sample variance~\cite{LMT09}:
\begin{align}
    \Tilde{\sigma}_\mathit{IS}^2 = \frac{1}{n - 1} \sum_{i=1}^{n} \Psi^2(y_{i}) L^2(y_{i}) - \frac{n}{n - 1} \overline{X}_\emph{IS}^2
\end{align}
By the central limit theorem, for large enough $n$, the distribution of $\sqrt{n}(\overline{X}_\emph{sim}-\mu)/\sigma$ converges to the standard normal distribution, which also applies to $\overline{X}_\emph{IS}$.
Therefore, the \ac{CI} for the target probability~$p_{\prop}(t)$  can be derived as: 
\begin{align}
	\left[\overline{X}_\emph{sim} - z_{\alpha/2}\cdot\frac{\Tilde{\sigma}_\emph{sim}}{\sqrt{n}}, \overline{X}_\emph{sim} + z_{\alpha/2}\cdot\frac{\Tilde{\sigma}_\emph{sim}}{\sqrt{n}}\right]\text{,}
\end{align}
where $z_{\alpha/2}$ corresponds to the $\alpha/2$-quantile of the standard normal distribution.

%% file: 03-developed-procedure.tex
\section{Switch from Stochastic State Classes to Simulation}
\label{sec:introducing-switch-sscs}
In this section, we provide an overview of our approach, resorting to an initial \ac{SSC} expansion and simulation afterwards (\Cref{sec:switch-overview}); we explain how an \ac{SSC} is conditioned on a transition firing to start a simulation offspring from it (\Cref{sec:conditioning-sscs}); and, we describe how samples are created from a conditioned \ac{SSC} (\Cref{sec:sampling-methods}).

\hide{This section elaborates on the proposed analysis method that resorts to an initial \ac{SSC} expansion and simulation afterwards. We describe the general concept (\Cref{sec:switch-overview}), explain how we determine a subzone of an original \acp{SSC} from which we start a simulation offspring (\Cref{sec:conditioning-sscs}) and finally cover how we create samples from a subzone of an \acp{SSC} (\Cref{sec:sampling-methods}).}


\subsection{Approach Overview}
\label{sec:switch-overview}
\markchange{We start with performing} an analytical expansion with \acp{SSC} using a predefined depth $d$, meaning that we only enumerate \acp{SSC}: \textit{i})~that can be reached at maximum after $d$ transition firings, and \textit{ii})~from which an \ac{SSC} satisfying target property~\prop can be reached with positive probability (which can be decided on the underlying \ac{SCG}---see \Cref{subsec:ssc}). 
From each \ac{SSC} $\Sigma$ at distance~$d$ from the root, for each enabled transition~$\gamma$ of~$\Sigma$ such that a state satisfying~\prop can be reached from the successor \ac{SSC} of $\Sigma$, we define a starting state for simulation offspring with pair $\langle\Sigma, \gamma\rangle$.
Let the \ac{RV} $Y_{\langle\Sigma, \gamma\rangle}$
\markchange{%
with support on $\mathbb{R}_{\geq 0}^{G+1}$ denote the age $\age$ and times-to-fire of the $G$ enabled transitions for a simulation offspring (from a $(G+1)$-dimensional \ac{SSC} $\Sigma$) when $\gamma$ fires first.
}
%
%
%
For given model time~$t$ and property~\prop, let the \ac{RV} $X_{\langle \Sigma, \gamma\rangle}$ with support on $\{0, 1\}$ denote the reward obtained when evaluating~\prop by performing a simulation run $\Psi_{\langle \Sigma, \gamma\rangle} : \mathbb{R}_{\geq 0}^{G+1} \times \mathcal{S} \rightarrow \{0, 1\}$ from pair $\langle\Sigma, \gamma\rangle$ until reaching time~$t$.
For simplicity, we omit the information of the simulation engine when referring to $\Psi_{\langle \Sigma, \gamma\rangle}$, so we write $X_{\langle \Sigma, \gamma\rangle} = \Psi_{\langle \Sigma, \gamma\rangle} (Y_{\langle\Sigma, \gamma\rangle})$.

\hide{In the developed approach, we perform at first an analytical expansion with \acp{SSC} using a predefined depth $d$, meaning that we only cover \acp{SSC} which have been reached at maximum after $d$ transition firings. From each \ac{SSC} $\Sigma$ with a distance of $d$ to the root, we define for each enabled transition $\gamma$ of $\Sigma$ a starting point for simulation offspring with a pair $\langle\Sigma, \gamma\rangle$.
We denote with the \ac{RV} $Y_{\langle\Sigma, \gamma\rangle} \in (\mathbb{R}_{0}^{+})^G$ the firing times of the enabled transitions for a simulation offspring when sampling from a $G$-dimensional \ac{SSC} $\Sigma$, under the condition that the transition $\gamma$ fires first. {\color{red}(Including the age variable, also synchronize the dimension of the SSC)} When employing \ac{IS}, the samples are associated with a likelihood, which is denoted with $L : (\mathbb{R}_{0}^{+})^G \rightarrow \mathbb{R}_{0}^{+}$.
For a predefined model time $t$, we can characterize the reward as \ac{RV} $X_{\langle \Sigma, \gamma\rangle} \in \{0, 1\}$ with $X_{\langle \Sigma, \gamma\rangle} = \Psi_{\langle \Sigma, \gamma\rangle} (Y_{\langle\Sigma, \gamma\rangle})$ outgoing from the pair $\langle\Sigma, \gamma\rangle$, where $\Psi_{\langle \Sigma, \gamma\rangle}$ corresponds to performing a simulation run until reaching model time $t$ and evaluating the target condition.}

\begin{figure}[t]
    \centering
    \includegraphics[width=\linewidth]{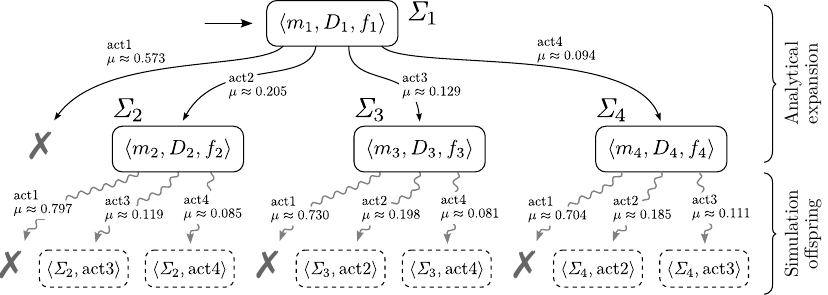}
    \caption{Illustration of the approach with analytical expansion depth $d=1$}
    \label{fig:illustration_approach}
    \vspace{-2ex}
\end{figure}

\markchange{%
Note that, other than belonging to the expolynomial class, no assumptions are made about the distributions of the \acp{RV}, nor on the presence of regeneration points.
This makes our approach applicable to a general class of stochastic processes, including e.g.\ \acp{MRP} and some non-Markovian systems.
}
For the \markchange{model}
in \Cref{fig:stpn} and expansion depth~$d=1$, \markchange{\Cref{fig:illustration_approach} illustrates the procedure:}
%
%
after the firing of transition \texttt{act1}, any further analysis can be omitted, since the target condition (i.e.,~that the activity 
\texttt{act1} is still ongoing while 
\texttt{act2}, \texttt{act3}, and \texttt{act4} have been completed) can never be reached after this event has happened.
%

\hide{For the example with four parallel overlapping activities in \Cref{fig:stpn} and expansion depth~$d=1$, an illustration of the developed procedure is given in \Cref{fig:illustration_approach}. As can be seen, any further analysis after the firing of the transition \texttt{act1} can be omitted, since the target condition can never be reached when this event has happened. This property can, e.g., be checked by running a reachability analysis upon the \ac{PN} by solely considering the markings and not any timing information.}

\subsection{Conditioning Stochastic State Classes on Fired Transition}
\label{sec:conditioning-sscs}

To define a starting state for a simulation offspring associated with pair $\langle\Sigma, \gamma\rangle$, the age variable~$\age$ and the times-to-fire of the enabled transitions are sampled from the joint PDF $f_{\langle\tau_\emph{age}, \vect{\tau}\rangle}$ of $\Sigma$ conditioned on transition~$\gamma$ firing first.
To this end, the times-to-fire of \ac{EXP} transitions can be handled independently of $\tau$ and the other times-to-fire, given that, due to the memoryless property, those \acp{RV} are independent of the other ones and, after a transition firing, each of them follows its respective \ac{EXP} distribution with the same rate~\cite{carnevali2008state}.
In detail, we only store the rates $\lambda_1, \ldots, \lambda_n$ of the involved \ac{EXP} distributions, which however can influence the evolution of the other times-to-fire, as discussed in the following.

$\Sigma$ is conditioned on transition~$\gamma$ firing first by the following steps (if \ac{EXP} transitions are not present, then steps~1, 2, and 3a are omitted and, in step~3b, variable $x_{exp}$ is not present and thus \Cref{eq:removeExp} does not need to be solved):

\hide{
The following steps are performed 
When assuming that the transition $\gamma$ has been fired first in $\Sigma$, the following steps have to be performed: {\color{red}(There can also be the case that there are no exponentially distributed transitions involved.)}}

\hide{To create the firing times $Y_{\langle\Sigma, \gamma\rangle}$ for an offspring, we need to cut out the part of the original density $f_{\langle\tau_\emph{age}, \vect{\tau}\rangle}$ from $\Sigma$ such that the resulting domain only contains those parts where the transition $\gamma$ is the first one to be fired. 
{\color{red} (Laura: say this is usually terned conditioned firing domain.)}
In this context, it is worth acknowledging that the firing times of exponentially distributed transitions {\color{red} (Laura: why not using acronym EXP?)} are handled independently of the other ones to save memory, as those variables are stochastically independent of the other variables and do not change in a firing step due to the memoryless property.
{\color{red} (Laura: add citation here.)}
In detail, we only store the rates $\lambda_1, ..., \lambda_n$ of the involved exponential distributions. However, exponentially distributed transitions can influence the timing behavior of the other transitions.

When assuming that the transition $\gamma$ has been fired first in $\Sigma$, the following steps have to be performed: {\color{red}(There can also be the case that there are no exponentially distributed transitions involved.)}}

\begin{enumerate}
    \item We calculate the aggregated \ac{EXP} distribution with rate $\lambda_{\emph{agg}} = \lambda_1 +  \ldots + \lambda_n$, being the rate of the firing time of any of these transitions (as the minimum of $n$ \ac{EXP} \acp{RV} with rates $\lambda_1, \ldots, \lambda_n$ is an \ac{EXP} \ac{RV} with rate $\lambda_1 + \ldots +  \lambda_n$).
	\item We add the aggregated \ac{EXP} time-to-fire to the joint \ac{PDF}:
    \begin{align}
        f_{\tau'}(x_\emph{age}, \vec{x}, x_\emph{exp}) = f_{\langle\tau_\emph{age}, \vect{\tau}\rangle}(x_\emph{age}, \vec{x}) \cdot e^{-\lambda_{\text{agg}}x_\emph{exp}}
    \end{align}
    where $\tau' = \langle\tau_\emph{age}, \vect{\tau}, \tau_\emph{exp}\rangle$ and $D'$ is the domain of $f_{\tau'}$.
	\item We distinguish whether the fired transition $\gamma$ is \ac{EXP} itself or not:
	\begin{enumerate}
		\item If yes, $x_\emph{exp}$ is at most the time-to-fire of each non-\ac{EXP} transition,
    		and it fires with probability~$ \mu_{\gamma}$, yielding the conditioned joint \ac{PDF} $f_{\tau''}$, with $\tau''=\langle \age', \vec{\tau}', \tau'_{exp} \rangle$ and $I=\{ i \,|\, i\neq age \wedge i \neq exp\}$:
            \begin{gather}
                \mu_{\gamma} = \frac{\lambda_{\gamma}}{\lambda_\emph{exp}} \int_{
                \begin{subarray}{l}
                \{ (x_\emph{age}, \vec{x}, x_\emph{exp}) \in D'\\
                \text{s.t. }
                x_\emph{exp} \leq x_i~\forall\,i \in I \}            
                \end{subarray}
                } 
                f_{\tau'}(x_\emph{age}, \vec{x}, x_\emph{exp}) dx_\emph{age} d\vec{x} dx_\emph{exp}\\
                f_{\tau''}(x_\emph{age}, \vec{x}, x_\emph{exp}) = \frac{\mathbf{1}_{
                \begin{subarray}{l}
                    \{ (x_\emph{age}, \vec{x}, x_\emph{exp}) \in D'\\ \text{s.t. } x_\emph{exp} \leq x_i~\forall\,i \in I \}
                \end{subarray}
                }
                (x_\emph{age}, \vect{x}, x_\emph{exp})}{\mu_{\gamma} \cdot \frac{\lambda_\emph{exp}}{\lambda_{\gamma}}} \cdot f_{\tau'}(x_\emph{age}, \vec{x}, x_\emph{exp})
            \end{gather}
            %
            
    	\item Otherwise (the fired transition $\gamma$ is not \ac{EXP}), then $x_\gamma$ must be at most $x_{exp}$ or the time-to-fire of every non-\ac{EXP} transition. It fires with probability~$\mu_\gamma$, yielding $f_{\tau''}$ with $ \tau'' = \langle \age', \vec{\tau}', \tau'_{exp} \rangle$ and $I=\{ i \,|\, i\neq age \wedge i\neq \gamma \}$: 		
        \begin{gather}
            \mu_{\gamma} = \int_{\{(x_\emph{age}, \vec{x}, x_\emph{exp})\in D' \mid x_{\gamma} \leq x_i~\forall\,i \in I \}} f_{\tau'}(x_\emph{age}, \vec{x}, x_\emph{exp}) dx_\emph{age} d\vec{x}dx_\emph{exp}\\
            f_{\tau''}(x_\emph{age}, \vec{x}, x_\emph{exp})
            = \frac{\mathbf{1}_{
            \begin{subarray}{l}
                 \{(x_\emph{age}, \vec{x}, x_\emph{exp})\in D' \\ \text{s.t. } x_{\gamma} \leq x_i~\forall\,i \in I \}
                 \end{subarray}
            }
            (x_\emph{age}, \vec{x}, x_\emph{exp})}{\mu_{\gamma}} f_{\tau'}(x_\emph{age}, \vec{x}, x_\emph{exp})
        \end{gather}
        As  \ac{EXP} times-to-fire are not affected by transition firings, $f_{\tau''}$ can be mar\-gi\-na\-lized with respect to $x_{exp}$, yielding $f_{\tau'''}$ with $ \tau''' = \langle  \age'', \vec{\tau}'' \rangle$:
        
        \begin{align}
            f_{\tau'''}(x_\emph{age}, \vec{x}) = \int_{0}^{\infty} f_{\tau''}(x_\emph{age}, \vec{x}, x_\emph{exp}) dx_\emph{exp}
            \label{eq:removeExp}
        \end{align}
	\end{enumerate}
\end{enumerate}

Given the firing probability~$\mu_\gamma$ of~$\gamma$ in \ac{SSC} $\Sigma$, we compute weight $w_{\langle\Sigma, \gamma\rangle} = \rho(\Sigma) \cdot \mu_{\gamma}$ where $ \rho(\Sigma)$ is the reaching probability of $\Sigma$, and $w_{\langle\Sigma, \gamma\rangle}$ comprises the probability of reaching the successor \ac{SSC} of $\Sigma$ through~$\gamma$ starting from $\Sigma$.
%
%
This  quantity is needed to compute the global reward $X$ by its individual rewards $X_{\langle \Sigma, \gamma\rangle}$, obtained from simulation offspring with pair $\langle \Sigma, \gamma \rangle$  (see  \Cref{sec:results-confidence-intervals}).

Besides the rates of \ac{EXP} transitions, we store the values of \ac{DET} transitions and times-to-fire of transitions with \ac{DET} time difference to the time-to-fire of another transition. Due to space limits, we refer to \cite{stochastic-state-classes,carnevali2008state} for evaluation of the firing probability~$\mu_\gamma$ and conditioned joint PDF with \ac{DET} transitions.

\hide{
Based on the firing probability for a transition $\gamma$, we can calculate the weight $w_{\langle\Sigma, \gamma\rangle} = \rho(\Sigma) \cdot \mu_{\gamma}$
%
%
which reflects the reaching probability of the corresponding \ac{SSC} with its outgoing transition. This is important for the calculation of the global reward $X$ by its individual rewards $X_{\langle \Sigma, \gamma\rangle}$ in Section~\ref{sec:results-confidence-intervals}.
Besides the rates of the \ac{EXP} transitions, we store the value of \ac{DET} transitions or a with deterministic time difference to another transition separately. 
Sampling for those transitions is straightforward. Without this setup, it would be necessary to represent the Dirac delta function in the joint PDF.}

\subsection{Sampling Methods}
\label{sec:sampling-methods}

\markchange{Given an \ac{SSC} conditioned on a transition firing, times-to-fire of \ac{EXP} transitions can be sampled individually, while inverse transform cannot be applied for their joint PDF $f_{\langle\tau_\emph{age}, \vect{\tau}\rangle}$ with $\age$, which 
takes a piece-wise expolynomial representation over a domain partition in \ac{DBM} sub-zones.
}
%
%
Therefore, we exploit two sampling methods that only need to evaluate $f_{\langle\tau_\emph{age}, \vect{\tau}\rangle}$ at certain points, namely the \ac{MH} algorithm~\cite{metropolis1953equation,hastings1970monte} and \ac{IS}~\cite{importance-sampling}. 
They operate differently (in particular, in contrast to \ac{IS}, \ac{MH} algorithm generates samples with the same likelihood) and thus yield different \ac{CI} evaluations in \Cref{sec:results-confidence-intervals}. We provide here a short description of these methods---for further details see \Cref{sec:algorithmic-desc-sampling}.

\paragraph{Metropolis-Hastings Algorithm.}
\label{sec:metropolis-hastings}
Starting from an arbitrary point in the \ac{PDF} domain, we iteratively sample a new point $\vect{x}'$ from a proposal \ac{PDF} based on the last point~$\vect{x}_{t}$ (e.g.,~normal \ac{PDF} centered at $\vect{x}_{t}$), accepting $\vect{x}'$ with probability $\min(\alpha, 1)$ where $\alpha = f(\vect{x}')/f(\vect{x}_t)$ (acceptance ratio). Assessed by measuring autocorrelation of samples with the Ljung-Box test~\cite{ljung-box}, we perform undersampling to ensure a low correlation between consecutive simulation offspring.

\paragraph{Importance Sampling.}
\label{sec:importance-sampling-ssc}
We introduce a proposal \ac{PDF} $\Tilde{f}$ with a known sampler, whose domain $\Tilde{D}$ contains the original domain $D$, such as
\begin{align}
    \Tilde{f}(\vec{x}) = \prod_{i=1}^{H} \mathbf{1}_{[l_i, u_i]}(x_i) \begin{cases}
        \frac{1}{u_i - l_i} & \text{if }l_i\neq -\infty \wedge u_i\neq \infty\\
        \lambda e^{-\lambda (x_i - l_i)} & \text{else if }u_i = \infty\\
        \lambda e^{-\lambda (u_i - x_i)} & \text{else if }l_i = -\infty
    \end{cases}\text{,}
    \label{eq:importance-sampling-proposal}
\end{align}
where $l_i$ and $u_i$ are the lower and upper bound, respectively, of the marginal domain of $x_i$ with $i \in \{1, ..., H\}$, where $H \leq G + 1$ is the number of non-EXP and non-DET transitions involved, and $\lambda \in \mathbb{R}^+$ is a parameter defined by the user.
Note that $\lambda$ has to be chosen carefully (e.g. neither too low nor too high) to avoid variance explosion~\cite{LMT09}.
Also note that, as $\tau_\emph{age}$ is encoded as the opposite of the elapsed time in the calculus of \acp{SSC}, we consider the case that $l_i=-\infty$.

%% file: 04-confidence-intervals.tex
\section{Collecting Results and Obtaining Confidence Intervals}
\label{sec:results-confidence-intervals}

In this section, we illustrate the derivation of \acp{CI}. 
We characterize the states that fulfill the target time-bounded transient property \prop in three cases:
\begin{itemize}
    \item 
    If \prop is satisfied by a state during simulation from pair $\langle \Sigma, \gamma\rangle$ (as described in \Cref{sec:introducing-switch-sscs}), we estimate the mean $\mu_{\langle \Sigma ,\gamma\rangle}$ of the reward with the sample mean
    \hide{The target can be covered during a simulation run employed from a pair $\langle \Sigma, \gamma\rangle$ as described in \Cref{sec:introducing-switch-sscs}. We can estimate the mean $\mu_{\langle \Sigma ,\gamma\rangle}$ of the reward from the pair $\langle \Sigma, \gamma\rangle$ with the sample mean}
    \begin{align}
        \overline{X}_{\langle \Sigma, \gamma\rangle} = \frac{1}{n_{\langle\Sigma, \gamma\rangle}} \cdot \sum_{i = 1}^{n_{\langle\Sigma,\gamma\rangle}} L(y_{\langle\Sigma, \gamma\rangle, i})\Psi(y_{\langle\Sigma, \gamma\rangle, i})
    \end{align}
    %
    %
    using simulation offspring times $y_{\langle\Sigma, \gamma\rangle, i}$, $i\in \{1, ..., n_{\langle\Sigma, \gamma\rangle}\}$, from $Y_{\langle\Sigma, \gamma\rangle}$. It converges, according to the central limit theorem, to a normal distribution:
    \begin{align}
        \overline{X}_{\langle \Sigma, \gamma\rangle} \rightarrow \mathcal{N}\left(\mu_{\langle\Sigma, \gamma\rangle}, \frac{\sigma_{\langle\Sigma, \gamma\rangle}^2}{n_{\langle\Sigma, \gamma\rangle}}\right)
        \label{eq:inidividual_sum_normal_distrib}
    \end{align}
    As described in \Cref{subsec:mcs}, if samples have the same weight (e.g., $L(y_{\langle\Sigma, \gamma\rangle, i)}) = 1$, which holds for the \ac{MH} algorithm), we estimate the variance with $\Tilde{\sigma}_{\langle\Sigma, \gamma\rangle}^2 = \overline{X}_{\langle\Sigma, \gamma\rangle}(1 - \overline{X}_{\langle\Sigma, \gamma\rangle})$. When incorporating weights of samples with \ac{IS}, we estimate the variance by the sample variance:
    \begin{align}
        \begin{split}
            \Tilde{\sigma}_{\langle \Sigma, \gamma\rangle}^2 &= \frac{1}{n_{\langle\Sigma, \gamma\rangle} - 1} \sum_{i=1}^{n_{\langle \Sigma, \gamma\rangle}} \Psi_{\langle \Sigma, \gamma\rangle}^2(y_{\langle \Sigma, \gamma\rangle, i}) L_{\langle \Sigma, \gamma\rangle}^2(y_{\langle \Sigma, \gamma\rangle, i})\\
            & \strut\quad - \frac{n_{\langle \Sigma, \gamma\rangle}}{n_{\langle \Sigma, \gamma\rangle} - 1} (\overline{X}_{\langle \Sigma, \gamma\rangle})^2
        \end{split}
        \label{eq:is-var-estimation}
    \end{align}
    \item If \prop is satisfied by \ac{SSC}~$\Sigma_\emph{det}$, then: 
    if \prop requires some property (marking condition) to be satisfied at time~$t$,
    the weight is the reaching probability $\rho(\Sigma_\emph{det})$ multiplied by the probability that $\Sigma_\emph{det}$ is the last \ac{SSC} reached at~$t$:
    \hide{Secondly, a target marking can be covered by an \ac{SSC} analytically, denoted with $\Sigma_\emph{det}$. The weight of this property is determined by the probability that a simulation offspring from the root state eventually leads up to this state multiplied by the probability that it currently stays in this \ac{SSC}:}
    \begin{align}
        w_{\Sigma_\emph{det}} = \rho(\Sigma_\emph{det}) \int_{
        \begin{subarray}{l}
        \{(x_\emph{age},\vec{x}) \in D \text{ s.t.}\\ 
        - x_\emph{age} \leq t \,\wedge\, x_i - x_\emph{age} > t\ \forall i\}
        \end{subarray}        
        } f_{\langle\tau_\emph{age}, \vec{\tau}\rangle}(x_\emph{age}, \vec{x}) dx_\emph{age}d\vec{x}
    \end{align}
    Otherwise, if \prop requires some property be satisfied within time~$t$, then $w_{\Sigma_\emph{det}}$ is derived as $\rho(\Sigma_\emph{det})$ multiplied by the probability that $\Sigma_\emph{det}$ is reached within~$t$:
    \begin{align}
        w_{\Sigma_\emph{det}} = \rho(\Sigma_\emph{det}) \int_{
        \{(x_\emph{age},\vec{x}) \in D \text{ s.t. } 
        - x_\emph{age} \leq t \}
        } f_{\langle\tau_\emph{age}, \vec{\tau}\rangle}(x_\emph{age}, \vec{x}) dx_\emph{age}d\vec{x}
    \end{align}

    \item If an \ac{SSC} satisfying \prop is unreachable from an \ac{SSC} (which can be decided using the \ac{SCG}, see \Cref{subsec:ssc}), \ac{SSC} expansion is stopped, yielding reward zero.
    
    \hide{Lastly, we can detect if the target is unreachable from an \ac{SSC}, e.g., by determining the distance from it in the corresponding \ac{SCG} as discussed in \Cref{subsec:ssc}. Whenever the distance to the target is infinite, we would never reach the target from this \ac{SSC}. As the reward of non-target states is zero, we can simply stop the exploration there.}
\end{itemize}
Combining the results of the first two cases, we estimate the total reward by \markchange{\Cref{eq:mean_mixed_analysis}, and then we can calculate the \ac{CI} for the estimator by \Cref{thm:CI}}:
\begin{align}
    \overline{X}_\text{mixed} = \sum_{\langle\Sigma, \gamma\rangle} w_{\langle\Sigma, \gamma\rangle} \cdot \overline{X}_{\langle\Sigma, \gamma\rangle} + \sum_{\Sigma_\emph{det}} w_{\Sigma_\emph{det}}
    \label{eq:mean_mixed_analysis}
\end{align}
%
\begin{theorem}[\Acl{CI} for composite analysis]\label{thm:CI}
    The $1 - \alpha$ \ac{CI} of the mean $\mu$ for the estimator in \Cref{eq:mean_mixed_analysis} can be calculated with
    \begin{align*}
        \left[\overline{X}_\text{mixed} - z_{\alpha/2}\cdot\sigma_\text{mixed}, \overline{X}_\text{mixed} + z_{\alpha/2}\cdot\sigma_\text{mixed}\right]\text{,}
    \end{align*}
    where $z_{\alpha/2}$ corresponds to the $\alpha/2$-quantile for the standard normal distribution $\mathcal{N}(0, 1)$ and $\sigma^2_\text{mixed}$ is defined as:
    \begin{align*}
    	\sigma^2_\text{mixed} = \sum_{\langle\Sigma, \gamma\rangle} \frac{w_{\langle\Sigma, \gamma\rangle}^2}{n_{\langle\Sigma, \gamma\rangle}}\cdot \sigma_{\langle\Sigma, \gamma\rangle}^2
    \end{align*}
\end{theorem}
\begin{proof}
    The estimator in \Cref{eq:mean_mixed_analysis} consists of the weighted sum of individual estimates for each combination of \ac{SSC} and outgoing transition. As denoted in \Cref{eq:inidividual_sum_normal_distrib}, 
    each individual sum converges to a normal distribution. \markchange{The deterministic values $w_{\Sigma_\emph{det}}$ can also be interpreted as $\mathcal{N}(w_{\Sigma_\emph{det}}, 0)$.} Furthermore, for two independent random variables $Y\sim \mathcal{N}(\mu_{Y}, \sigma_{Y}^2)$ and $Z\sim \mathcal{N}(\mu_{Z}, \sigma_{Z}^2)$, the sum $Y + Z$ is distributed according to $\mathcal{N}(\mu_{Y} + \mu_{Z}, \sigma_{Y}^2 + \sigma_{Z}^2)$. Besides, $\text{Var}(aY) = a^2\text{Var}(Y)$ holds for the variance of a random variable $Y$ with constant $a\in \mathbb{R}_{\geq 0}$.
    Using this information, we obtain that the estimator in \Cref{eq:mean_mixed_analysis} converges to a normal distribution $\mathcal{N}(\mu, \sigma^2_\text{mixed})$ with:
    \begin{align*}
    	\overline{X}_\emph{mixed} \rightarrow \mathcal{N}\left(\sum_{\langle\Sigma, \gamma\rangle} w_{\langle\Sigma, \gamma\rangle}\cdot \mu_{\langle\Sigma, \gamma\rangle} + \sum_{\Sigma_\emph{det}} w_{\Sigma_\emph{det}} ~\boldsymbol{,} \sum_{\langle\Sigma, \gamma\rangle} \frac{w_{\langle\Sigma, \gamma\rangle}^2}{n_{\langle\Sigma, \gamma\rangle}}\cdot \sigma_{\langle\Sigma, \gamma\rangle}^2 \right)
    \end{align*}
    The expression $(\overline{X}_\text{mixed} - \mu)/\sigma_\text{mixed}$ then converges to a standard normal distribution $\mathcal{N}(0, 1)$ for a sufficiently large sample size. Thus, we can rearrange the formula and calculate the \ac{CI} as usual.\hfill$\square$
\end{proof}

%% file: 05-experimental-evaluation.tex
\section{Experimental Evaluation}
\label{sec:experimental-evaluation}

In this section, we evaluate the approach with the example of \Cref{fig:stpn} (\Cref{subsec:example-4-servers}) and a \ac{DFT} case study (\Cref{subsec:example-dft}).
A prototype implementation was developed using the Sirio library~\cite{sirio-github-code},
\markchange{%
which supports \acp{STPN} and \acp{SSC}, and enables experimentation on general stochastic processes---see \Cref{subsec:stpn,sec:switch-overview}.}
Experiments were performed on an Apple M1 CPU with 16 GB RAM.

\hide{
We are evaluating the developed approach with two scenarios, the example illustrated in \Cref{fig:stpn} in \Cref{subsec:example-4-servers} as well as a \ac{DFT} in \Cref{subsec:example-dft}.
The implementation of the prototype is done with the help of the SIRIO library~\cite{sirio-github-code}, which already supports \acp{STPN} and \acp{SSC} (see \Cref{subsec:stpn}). As evaluation machine, we utilize an Apple M1 CPU with 16 GB RAM.
}

\subsection{Four Overlapping Activities}
\label{subsec:example-4-servers}

\begin{figure}[t]
    \centering
    \includegraphics{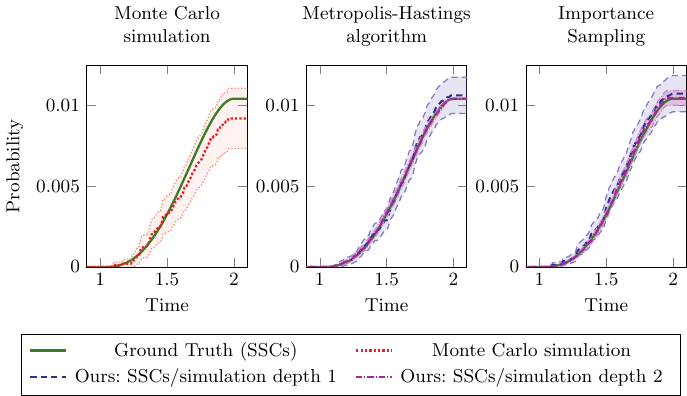}
    \caption{Transient probabilities for four overlapping activities with $95\%$ \acp{CI}}
    \label{fig:plot_parallel_failing_4_servers}
\end{figure}

The first scenario is the model with four parallel overlapping activities from \Cref{fig:stpn}, for which \Cref{fig:illustration_approach} illustrated the mixture of \acp{SSC} analysis and simulation offspring.
Although this model can be fully analyzed analytically, we use it here to investigate the effects of different \ac{SSC} expansion depths and the two sampling methods.
We study the probability that the system fails up to some model time $t$, i.e.\ the PCTL-like property ${P\left( F_{\leq t}~\mathtt{Target}=1 \right)}$.
\Cref{fig:plot_parallel_failing_4_servers} depicts the transient probabilities for this model when using different analysis techniques:
\begin{enumerate}
    \item \textcolor{Red!90!Black}{\textbf{\ac{MC} simulation}}: 10000 simulation runs
          $\rightarrow$\hfill 1085 ms analysis runtime;
    \item \textcolor{Blue}{\textbf{Ours: \acp{SSC}/simulation depth 1}}: 500 simulation offspring per SSC/tran\-sition pair
          $\rightarrow$\hfill\mbox{14476 ms (\ac{MH} algorithm)/702 ms (\ac{IS}) runtime};
    \item \textcolor{Mulberry}{\textbf{Ours: \acp{SSC}/simulation depth 2}}: 200 simulation offspring per SSC/tran\-sition pair
          $\rightarrow$\hfill\mbox{5290 ms (\ac{MH} algorithm)/800 ms (\ac{IS}) runtime};
    \item \textcolor{OliveGreen}{\textbf{Ground Truth}}: build complete \ac{SSC} tree
          $\rightarrow$\hfill \mbox{809 ms analysis runtime}.
\end{enumerate}
As can be seen, regular \ac{MC} simulation shows the widest \ac{CI} and the largest deviation from the Ground Truth, while the two other methods successively narrow the width of the \acp{CI}. In the end, the second tested method performs in total $(3\cdot 2)\cdot 500 = 3000$ simulation runs, while the third one does $(3\cdot 2\cdot 1) \cdot 200 = 1200$ simulation runs. Thus, both methods resulted in a higher accuracy while performing fewer simulation runs than crude \ac{MC} simulation.
When comparing the \ac{MH} algorithm and \ac{IS} for creating simulation offspring, we notice a high runtime difference between the two methods. This stems from the high undersampling step of $800$ needed for \ac{MH} according to the Ljung-Box test to obtain uncorrelated simulation offspring (see \Cref{sec:algorithmic-desc-sampling}).
Furthermore, with \ac{SSC} expansion depth $d=2$, sampling with \ac{MH} results in a zero-variance estimate for $t\geq 2$, while we observe non-zero variance when resorting to \ac{IS}. The reason for this is that for \ac{MH} we can estimate the variance with $\Tilde{\sigma}_{\langle\Sigma, \gamma\rangle}^2 = \overline{X}_{\langle\Sigma, \gamma\rangle}(1 - \overline{X}_{\langle\Sigma, \gamma\rangle})$, while for \ac{IS} we resort to \Cref{eq:is-var-estimation}. Since every simulation employed from depth $d=2$ always reaches the target state, the corresponding mean equals~$1$. Thus, the estimator for \ac{MH} returns variance zero, whereas we obtain for \ac{IS} non-zero variance due to the different likelihood of the samples.


\subsection{Repairable Dynamic Fault Tree}
\label{subsec:example-dft}

For a more complex example, we study a repairable \ac{DFT} following the semantics of \cite{MBD20}.  
Our model uses the AND, PAND, and SPARE gates shown in \Cref{fig:DFT,code:DFT}, as well as three repair boxes where the repair priority of the left one is first UPS, then AC.\!%
\footnote{\label{ftm:extendedDFT}%
The children of a SPARE in \cite{MBD20} are \acp{BE} or spare \acp{BE}. We extend this for cold SPAREs \cite{JKSV18}, with an inactive \emph{spare subtree}---the right AND gate---that cannot fail while the \emph{primary subtree}---the left PAND gate---is operational.
Repairs of the primary subtree reset the spare subtree state (cf.\ spare \ac{BE} semantics \cite{MBD20,budde-non-markovian-repairable-ft}).
}
This models a highly reliable system powered by an unreliable grid, which has a UPS to remain operational during the recurrent blackouts.
The UPS battery is replaced periodically: if during the replacement a blackout occurs, an emergency mechanism turns on two diesel generators.
If both generators fail during the blackout, and before the UPS battery has been replaced, a system failure occurs.
We measure the probability of system failure before 21 time units elapse, i.e.\ the transient PCTL-like property ${P\left( F_{\leq21}~\mathtt{Target}=1 \right)}$.

\begin{figure}[t]
  \centering
  \begin{minipage}[b]{.35\linewidth}
    \includegraphics[width=\linewidth]{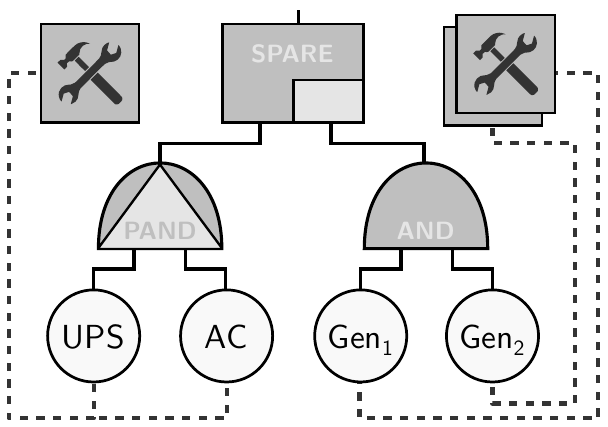}
    \caption{Repairable \ac{DFT}\cref{ftm:extendedDFT}}
    \label{fig:DFT}
  \end{minipage}
  \hfill
  \begin{minipage}[b]{.6\linewidth}
    \def\dist{\raisebox{-.7ex}{\textasciitilde}}
    \begin{lstlisting}[%
        style=Kepler,
        caption= \ac{DFT} from \Cref{fig:DFT} in Kepler syntax \cite{budde-non-markovian-repairable-ft},
        label=code:DFT,
        belowskip=-1ex,
        abovecaptionskip=\medskipamount,
        %float,
    ]
toplevel "Target";
"Target" spare "PAND" "AND";
"PAND"   pand  "UPS"  "AC";
"AND"    and   "Gen1" "Gen2";
"UPS"  fail`\dist`uni(^9.95,12^)        repair`\dist`exp(^5^);
"AC"   fail`\dist`uni(^18,20^)          repair`\dist`uni(^0,0.5^);
"Gen1" fail`\dist`exp(^1^) dorm`\dist`dir($\infty$) repair`\dist`uni(^1,2^);
"Gen2" fail`\dist`exp(^2^) dorm`\dist`dir($\infty$) repair`\dist`uni(^2,4^);
"R_PAND" rbox prio "UPS" "AC";
"R_GEN1" rbox prio "Gen1";
"R_GEN2" rbox prio "Gen2";
    \end{lstlisting}
  \end{minipage}
\end{figure}

\hide{%
The critical part of the analysis of this system is to detect the failure of the PAND gate, as this is only caused in the time interval when the failing time for UPS is twice near the lower bound, while the failing time for AC is near the upper bound. This causes regular \ac{MC} simulation to fail because of the low reaching probability, the same holds for \ac{ISPLIT}, at least if the \ac{IF} is solely based on the marking without incorporating any timing information.
\ac{IS} can work, but requires human insight to adjust the \acp{PDF} of the failure times of UPS and AC, and functional change of measures are hard to achieve. The analytical expansion with \acp{SSC}, however, can cover the PAND analytically and resort to simulation for the generators. In detail, we perform the following experiments, whose results can be seen in \Cref{fig:plot_dft}:
}

This \ac{DFT} has two failure modes, with the AND modeling a race condition of \ac{EXP} random variables and uniformly distributed revert transitions---\emph{repairs}---analyzable, e.g.,\ via \ac{MC} or \ac{IS}.
In contrast, PAND failures require UPS to fail before AC, and are reverted when AC is repaired.
The failure and repair distributions of the \acp{BE} make PAND failures a rare event that rules out \ac{MC} analysis.
Time-agnostic \ac{ISPLIT} approaches such as \cite{budde-non-markovian-repairable-ft} are equally impractical.
\ac{IS} could work, but requires non-trivial proposal failure \acp{PDF} of UPS and AC, which result in more frequent failures in the proper order.
Instead, expansions with \acp{SSC} can cover the PAND analytically, and afterwards, one can resort to \ac{MC} or \ac{IS} to study AND failures.
We experiment with these different approaches as follows:
\begin{enumerate}
    \item \textcolor{JungleGreen}{\textbf{Importance Sampling}}: We set up a mixture of \acp{PDF} for the failure of UPS and AC, to increase the probability of their ordered failure before repair:
    \vspace{-1.5ex}  
    \begin{align*}
        \text{UPS} \sim \begin{cases}
            \text{Unif}(9.95, 10) & \text{with 50\%},\\
            \text{Unif}(10, 12)   & \text{with 50\%};
        \end{cases} \qquad
        \text{AC} \sim \begin{cases}
            \text{Unif}(18, 19.9) & \text{with 50\%},\\
            \text{Unif}(19.9, 20) & \text{with 50\%}.
        \end{cases} \\[-5ex]  
    \end{align*}
    We perform 50000 simulation runs and observe a runtime of 44.7 s. 
    \item \textcolor{Blue}{\textbf{Ours: SSC/simulation depth 5}}: We perform 10000 simulation offspring per \acp{SSC}/transition pair, observing a runtime of 7.8 s. 
    \item \textcolor{Mulberry}{\textbf{Ours: SSC/simulation depth 12}}: 
    As above, with 41.0 s of runtime. 
\end{enumerate}

\begin{figure}[t]
    \centering
    \includegraphics{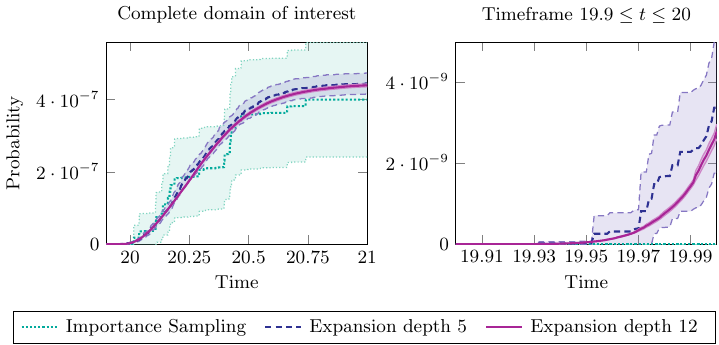}
    \caption{Transient probabilities for repairable \ac{DFT} example with $95\%$ \acp{CI}}
    \label{fig:plot_dft}
    \tikz[overlay]{
        \node (LOG) at (3.2,4.5) {\includegraphics[width=.19\linewidth]{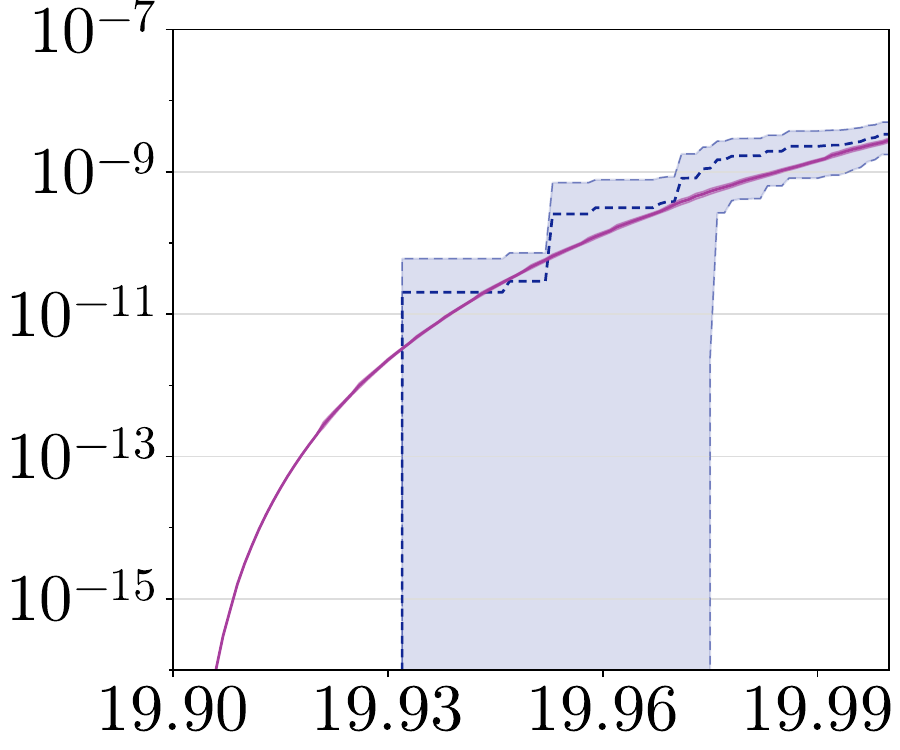}};
        \node [above=-7pt of LOG] {\smaller[2]~~~~logscale $y$-axis};
    }
\end{figure}

In our combined \ac{SSC}/simulation approach we use only \ac{IS} for creating simulation offspring, as the need for a high undersampling step renders \ac{MH} impractical.
\ac{SSC} expansions of depth $\geqslant 5$ suffice to analytically cover the PAND gate failure---i.e.\ when the second UPS failure occurs before the AC failure---resulting in a useful mixture of \acp{SSC} and simulation.
However, too-high expansion depths beyond $12$ lead to a runtime explosion for \ac{SSC} analysis, mainly due to the complexity stemming from splitting the analytical representation into subdomains.

\Cref{fig:plot_dft} shows how our mixed \ac{SSC}/simulation analysis results in a higher accuracy than \ac{IS}.
Moreover, applying \ac{IS} to this example required non-trivial human insight, and the beneficial (or detrimental) effect of the proposed \acp{PDF} is not immediately clear, as opposed to the choice of the \ac{SSC} expansion depth.
In that last respect, expanding up to depth $12$ results in a (partial) analytical coverage of the target states, which is why the mixed SSC/simulation approach can detect probabilities below $10^{-14}$ near $t = 19.9$---see the right plot in \Cref{fig:plot_dft}.

%% file: 06-conclusion.tex
\section{Conclusion}\label{sec:conclusion}

In this paper, we presented an approach to compute transient probabilities by first performing an analytical expansion with \acp{SSC} and then resorting to simulation. 
We presented two different solutions to create simulation offspring, \markchange{namely the \ac{MH} algorithm and the \ac{IS} method, where the latter seems to be preferable due to the fact that it does not require undersampling}.
\markchange{For both methods, we} demonstrated how to derive confidence intervals. 
Our prototype was evaluated with two examples: (a) four overlapping parallel activities and (b) a repairable \ac{DFT}, where the latter demonstrated how this composite analysis method can outperform classic \ac{RES} techniques such as \ac{IS} and \ac{ISPLIT}.

The main drawback of the method for the applicability in \ac{RES} is that it requires the critical event (that makes standard simulation impractical) to occur near the root of the state space.
We see two ways to overcome this limitation.
First, one could introduce multilevel switching from simulation to \acp{SSC}, similarly to \ac{ISPLIT}, whenever a critical region is encountered.
The challenging part is to detect when such a region occurs.
Secondly---for specific model classes---we can exploit regenerations and run our approach for each regenerative epoch.

Furthermore, sensitivity studies would help to quantify the effect of the number of samples on the tradeoff between accuracy and time gain, e.g.\ for the example in \Cref{subsec:example-dft}.
In this context, one challenge is to design more robust methods for deriving \acp{CI}. When dealing with a low probability or a small sample size, the target might never hit and we falsely assume a zero-variance with the variance estimators, e.g. $\overline{\sigma}_\emph{sim}^2 = \overline{X}_\emph{sim} (1 - \overline{X}_\emph{sim})$ when using the \ac{MH} algorithm to create simulation offspring.
Analyses via standard \ac{MC} simulation can use the Wilson score interval in such situations, as it provides better \ac{CI} coverage~\cite{wilson-interval}---we envision that similar techniques might be applicable to the combined \acp{SSC}/simulation approach as well.

%% file: 95-appendix.tex
\appendix

\section{Stochastic Time Petri Net of Dynamic Fault Tree}
\label{sec:pn-dft-example}

This section describes the \ac{STPN} crafted for the repairable \ac{DFT} example analyzed in \Cref{subsec:example-dft}, which is illustrated in \Cref{fig:stpn-dft}.
\begin{figure}[b]
    \centering
    \includegraphics[width=\linewidth]{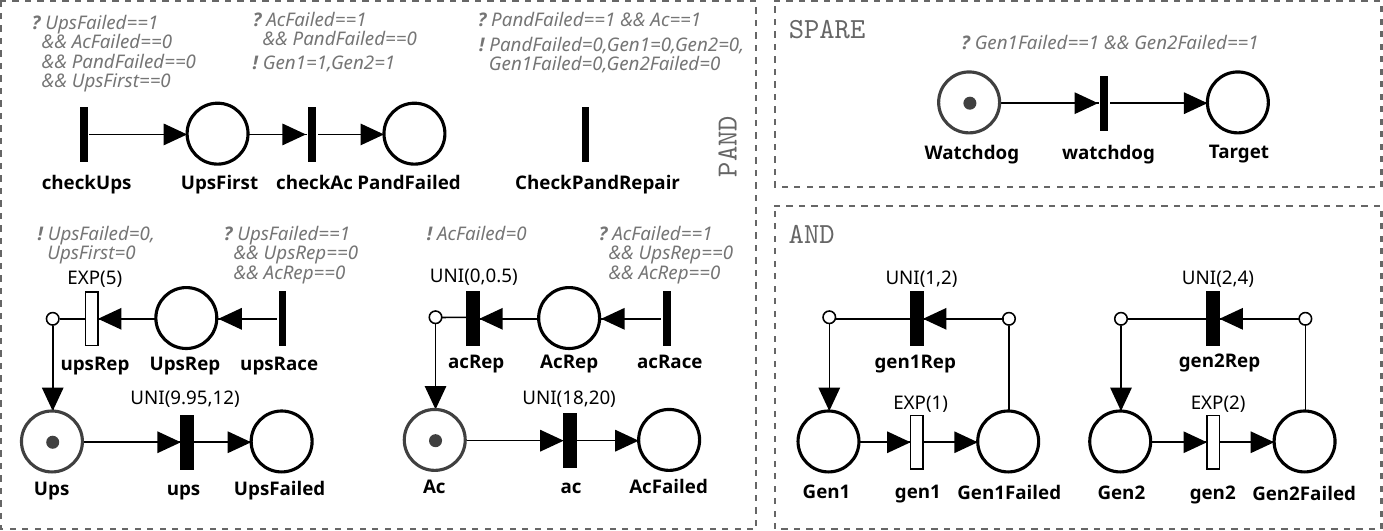}
    \caption{\ac{STPN} of the \ac{DFT} shown in \Cref{fig:DFT}.
    }
    \label{fig:stpn-dft}
\end{figure}
The initial marking of the \ac{STPN} is $\texttt{Ups} = 1, \texttt{Ac} = 1$ with enabled transitions \texttt{ups} and \texttt{ac}, which characterize the failure times of the battery and the power system, respectively. Since both components share the same repair box, we need to implement a logic that ensures the mutual exclusion of the repairs. This is done using the transitions \texttt{upsRace} and \texttt{acRace}, which marks the beginning of a repair process by adding a token to \texttt{UpsRep} and \texttt{AcRep}, respectively. Both \texttt{upsRace} and \texttt{acRace} can only fire when there is no repair process ongoing (which occurs when places \texttt{AcRep} and \texttt{UpsRep} are empty, respectively). The repair itself is modeled by transitions \texttt{upsRep} and \texttt{acRep}, which removes the token from the place representing the corresponding failure (\texttt{UpsFailed} and \texttt{AcFailed}, respectively), ending the repair process (represented by the firing of transitions \texttt{upsRep} and \texttt{acRep}, respectively), and updating the failure conditions of the PAND gate (modeled by places \texttt{UpsFirst} and \texttt{PandFailed}).
Transition \texttt{checkUps} fires when the power system has failed, but not the battery, which ensures the failure order of the PAND gate.
The activation of both diesel generators is represented by the firing of \texttt{checkAc}, which adds a token to \texttt{Gen1} and \texttt{Gen2}.
The repair of the PAND is modeled by transition \texttt{CheckPandRepair}, which fires after a failure of the PAND gate as soon as the battery is repaired, and switches off both diesel generators (by removing tokens from places \texttt{Gen1}, \texttt{Gen2}, \texttt{Gen1Failed}, and \texttt{Gen2Failed}).
Similarly to the power system and the battery, the failure of diesel generators is modeled by transitions \texttt{gen1} and \texttt{gen2}, respectively.
Since both diesel generators have an individual repair box, no synchronization is necessary, and thus transitions \texttt{gen1Rep} and \texttt{gen2Rep} suffice.
Finally, the transition \texttt{watchdog} keeps track of the failure condition that both diesel generators fail and thus finally the SPARE gate, and adds a token to place \texttt{Target} modeling a system failure.


\section{Implementation Details of Sampling Methods}
\label{sec:algorithmic-desc-sampling}

We provide implementation details and parameters used for the sampling methods for creating simulation offspring from \acp{SSC}.

For the \ac{MH} algorithm, we begin with an arbitrary starting point inside the domain $D$ of the \ac{SSC} $\Sigma$. We obtain this point by uniform sampling from the multidimensional rectangle $\prod_{i=1}^{H} [l_i, u_i]$, where $l_i$ and $u_i$ are the minimum and maximum possible value of a variable, and checking if the obtained point lies inside $D$. When we are dealing with an infinite domain, we cap the lower and upper bound to a predefined value (however, this is not relevant here, as this case does not happen in the example presented in \Cref{fig:stpn}).
In each step, we predict a new candidate $x'$ with the help of a proposal \ac{PDF} $g(x\mid y)$.
We can calculate the acceptance ratio based on the evaluated \ac{PDF} of the new candidate $x'$ and the old point $x_{t}$. The acceptance ratio $\alpha$ is used to determine if the new point will be accepted or not.
In our implementation, we use for $g(x\mid y)$ a separate normal distribution for each dimension that is centered around $y$ with variance $\sigma$. To determine $\sigma$, we start with $\sigma = 1$ and perform a binary search-like warm-up procedure to ensure that $\sigma$ is configured in a way such that the acceptance ratio lies on average between $0.2$ and $0.3$, which is compliant with the regular rule-of-thumb acceptance ratio of $\approx 0.234$~\cite{mh-acceptance-ratio}.
%
%
In detail, we perform $100$ rounds, where in each round we perform $100$ \ac{MH} steps and calculate the mean acceptance ratio. Depending on whether the acceptance ratio lies below $0.2$ or above $0.3$, we divide or multiply $\sigma$ by $1.25$ or leave $\sigma$ unchanged.
Furthermore, we determine an undersampling step for the obtained samples to ensure that the obtained samples are uncorrelated. We start with an undersampling step of $100$, assess the autocorrelation with the multivariate Ljung-Box test~\cite{ljung-box}, and increase the undersampling step incrementally by $100$. In detail, we resorted to the Microsoft WPA library for R\footnote{See function documentation at: \url{https://microsoft.github.io/wpa/reference/LjungBox.html}}, using $10^5$ samples for each simulation start point when using an analytical expansion depth of $d=1$ in the example of four overlapping activities shown in \Cref{fig:stpn}. 
Results indicate that an undersampling step of $800$ is necessary to obtain uncorrelated samples.

For \ac{IS}, as the variables of the proposal density in \Cref{eq:importance-sampling-proposal} are stochastically independent, we handle each variable individually. Thereby, we calculate the likelihood for each variable by comparing the marginal densities of the proposal density and the true density and multiply these values to obtain the global likelihood. We reject samples that are only part of the proposal domain and not the original one to obtain only samples with non-zero likelihood. This makes it necessary to scale the likelihood of each sample afterwards by the probability that a sampled point from the proposal \ac{PDF} $\Tilde{f}$ lies in the original domain $D$, which can be calculated with $\int_{\vec{x} \in D} \Tilde{f}(\vec{x}) d\vec{x}$.
We resorted to $\lambda = 1$, as the repairable \ac{DFT} example in \Cref{fig:DFT} deals with \ac{EXP} transitions with rate $1$ or $2$.

%% file: acronyms.tex

\begin{acronym}[ABCDEFGHIJK]
    \acro{AP}{atomic proposition}
    \acro{BE}[BE]{Basic Element}
    \acro{BFS}[BFS]{breadth-first search}
    \acro{BM}[BM]{Bernstein mixed polynomial and expolynomial}
    \acro{BP}[BP]{Bernstein polynomial}
    \acro{CDF}[CDF]{Cumulative Distribution Function}
    \acro{CI}[CI]{Confidence Interval}
    \acro{CSL}{continuous stochastic logic}
    \acro{CTMC}[CTMC]{continuous time Markov chain}
    \acro{DET}[DET]{deterministic}
    \acro{DBM}[DBM]{Difference Bounds Matrix}
    \acro{DFT}[DFT]{Dynamic Fault Tree}
    \acro{DSPN}[DSPN]{deterministic and stochastic Petri net}
    \acro{DTMC}[DTMC]{discrete time markov chain}
    \acro{EXP}[EXP]{exponential}    
    \acro{ERTMS}[ERTMS]{European Rail Traffic Management System}
    \acro{ES}[ES]{expected success}
    \acro{ETCS}[ETCS]{European Train Control System}
    \acro{FE}[FE]{fixed effort}
    \acro{FS}[FS]{fixed success}
    \acro{FTA}[FTA]{fault tree analysis}
    \acro{FT}[FT]{fault tree}
    \acro{FW}[FW]{Floyd-Warshall}
    \acro{GEN}[GEN]{general}
    \acro{GSPN}[GSPN]{generalized stochastic Petri net}
    \acro{IF}[IF]{importance function}
    \acro{IMM}[IMM]{immediate}
    \acro{IOSA}[IOSA]{Input/Output Stochastic Automata}
    \acro{IS}[IS]{Importance Sampling}
    \acro{ISPLIT}[ISPLIT]{Importance Splitting}
    \acro{JPD}[JPD]{joint probability distribution}
    \acro{MCMC}{markov chain Monte Carlo}
    \acro{MC}[MC]{Monte Carlo}
    \acro{MH}{Metropolis-Hastings}
    \acro{MRP}[MRP]{Markov Regenerative Process}
    \acro{MRnP}[MRnP]{Markov renewal process}
    \acro{MSE}[MSE]{mean-squared error}
    \acro{NDA}{non-deterministic analysis}
    \acro{PDF}[PDF]{Probability Density Function}
    \acro{PN}[PN]{Petri net}
    \acro{RESTART}[RESTART]{repetitive simulation trials after reaching thresholds}
    \acro{RES}[RES]{Rare Event Simulation}
    \acro{SBE}[SBE]{Spare Basic Element}
    \acro{SC}[SC]{state class}
    \acro{SCG}[SCG]{State Class Graph}
    \acro{SEQ}[SEQ]{sequential Monte Carlo}
    \acro{SMC}[SMC]{Statistical Model Checking}
    \acro{SSC}[SSC]{Stochastic State Class}
    \acro{STA}[STA]{stochastic timed automaton}
    \acro{STPN}[STPN]{Stochastic Time Petri Net}
    \acro{TA}[TA]{timed automaton}
    \acro{TPN}[TPN]{Time Petri Net}
    \acro{VI}[VI]{Value Iteration}
    \acro{i.i.d.}[i.i.d.]{independent and identically distributed}
    \acro{w.l.o.g.}[w.l.o.g.]{without loss of generality}
    \acro{RV}[RV]{Random Variable}
    \acroplural{BE}[BEs]{Basic Elements}
    \acroplural{BM}[BMs]{Bernstein mixed polynomials and expolynomials}
    \acroplural{BP}[BPs]{Bernstein polynomials}
    \acroplural{CI}[CIs]{Confidence Intervals}
    \acroplural{CTMC}[CTMCs]{continuous time Markov chains}
    \acroplural{DBM}[DBMs]{difference bound matrices}
    \acroplural{DTMC}[DTMCs]{discrete time Markov chains}
    \acroplural{FT}[FTs]{fault trees}
    \acroplural{GSPN}[GSPNs]{generalized stochastic Petri nets}
    \acroplural{MRP}[MRPs]{Markov Regenerative Processes}
    \acroplural{MRnP}[MRnPs]{Markov renewal processes}
    \acroplural{PN}[PNs]{Petri Nets}
    \acroplural{RV}[RVs]{Random Variables}
    \acroplural{SC}[SCs]{State Classes}
    \acroplural{SMP}[SMPs]{Semi-Markov Processes}
    \acroplural{SSC}[SSCs]{Stochastic State Classes}
    \acroplural{STA}[STAs]{stochastic timed automata}
    \acroplural{STPN}[STPNs]{Stochastic Time Petri Nets}
    \acroplural{TA}[TAs]{timed automata}
    \acroplural{TPN}[TPNs]{Time Petri Nets}
\end{acronym}